\documentclass[sigconf, nonacm]{acmart}

\usepackage[utf8]{inputenc} % allow utf-8 input
\usepackage[T1]{fontenc}    % use 8-bit T1 fonts
\usepackage{hyperref}       % hyperlinks
\usepackage{url}            % simple URL typesetting
\usepackage{booktabs}       % professional-quality tables
\usepackage{amsfonts}       % blackboard math symbols
\usepackage{nicefrac}       % compact symbols for 1/2, etc.
\usepackage{microtype}      % microtypography
\usepackage{float}
\usepackage{algorithmic}
\usepackage{algorithm, setspace}
\usepackage{graphicx}
\usepackage{amsmath}
\usepackage{amsfonts}
\usepackage{svg}
\usepackage{subcaption}
\usepackage{soul}
\usepackage{ulem}
\usepackage{xcolor}
\usepackage{listings}
\usepackage{balance}
\usepackage{enumitem}
\usepackage{multirow}

\newtheorem{definition}{Definition}
\newtheorem{theorem}{Theorem}

\newtheorem{lemma}{Lemma}

\newcommand{\longname}{IDentity with Locality Hash}
\newcommand{\name}{IDL }
\newcommand{\namens}{IDL}

\begin{document}
\title{IDentity with Locality: An ideal hash for gene sequence search}

%%
%% The "author" command and its associated commands are used to define the authors and their affiliations.

\author{Aditya Desai}
\authornotemark[1]
\affiliation{%
  \institution{Rice University}
  \streetaddress{6100 Main St}
  \city{Houston}
  \state{Texas}
  \postcode{77005}
}
\email{apd10@rice.edu}

\author{Gaurav Gupta}
\authornotemark[1]
\affiliation{%
  \institution{Rice University}
  \streetaddress{6100 Main St}
  \city{Houston}
  \state{Texas}
  \postcode{77005}
}
\email{gaurav.gupta@rice.edu}

\author{Tianyi Zhang}
\authornote{Equal contribution}
\affiliation{%
  \institution{Rice University}
  \streetaddress{6100 Main St}
  \city{Houston}
  \state{Texas}
  \postcode{77005}
}
\email{tz21@rice.edu}

\author{Anshumali Shrivastava}
\affiliation{%
  \institution{Rice University, ThirdAI Corp.}
  \streetaddress{6100 Main St}
  \city{Houston}
  \state{Texas}
  \postcode{77005}
}
\email{as143@rice.edu}

%%
%% The abstract is a short summary of the work to be presented in the
%% article.
\begin{abstract}
Gene sequence search is a fundamental operation in computational genomics with broad applications in medicine, evolutionary biology, metagenomics, and more. Due to the petabyte scale of genome archives, most gene search systems now use hashing-based data structures such as \textit{Bloom Filters} (BF). The state-of-the-art systems such as \textit{Compact bit-slicing signature index} (COBS) \cite{Bingmann2019COBSAC} and \textit{Repeated And Merged Bloom filters} (RAMBO) \cite{gupta2021fast} use BF with Random Hash (RH) functions for gene representation and identification. The standard recipe is to cast the gene search problem as a sequence of membership problems testing if each subsequent gene substring (called kmer) of $Q$ is present in the set of kmers of the entire gene database $D$. We observe that RH functions, which are crucial to the memory and the computational advantage of BF, are also detrimental to the system performance of gene-search systems. While subsequent kmers being queried are likely very similar, RH, oblivious to any similarity, uniformly distributes the kmers to different parts of potentially large BF, thus triggering excessive cache misses and causing system slowdown.

We propose a novel hash function called the Identity with Locality (IDL) hash family, which co-locates the keys close in input space without causing collisions. This approach ensures both cache locality and key preservation. IDL functions can be a drop-in replacement for RH functions and help improve the performance of information retrieval systems. We give a simple but practical construction of IDL function families and show that replacing the RH with IDL functions reduces cache misses by a factor of $5\times$, thus improving query and indexing times of SOTA methods such as COBS and RAMBO by factors up to $2\times$ without compromising their quality. We also provide a theoretical analysis of the false positive rate of BF with IDL functions. Our hash function is the first study that bridges \textit{Locality Sensitive Hash} (LSH) and RH to obtain cache efficiency. Our design and analysis could be of independent theoretical interest.
\end{abstract}

\maketitle

\section{Introduction}

Data mining systems such as genome indices heavily rely on hash functions~\cite{practicalML,bigSI}: functions that map arbitrary size values to a fixed range $[m]=\{0,\dots,{m-1}\}$. For instance, some fundamental data structures used for search and estimation problems are hash tables, \textit{Bloom filters}~\cite{bloom1970space}, \textit{count-sketches}~\cite{cormode2005improved}, etc., which use Random Hash (RH) functions and their variants~\cite{universalHashcarter1979, murmur, xxHash}. While RH and other hash functions have revolutionized search systems~\cite{ioffe2010improved,IndykAndMotwani, wang2017flash, gupta2022bliss, baeza1999modern, bigSI, elworth2020petabytes,practicalML, ryanDeepLeanring} by allowing us to build probabilistic data structures that exponentially improve memory and computational requirements, they are not conducive for system performance. Search systems often serve a batch of queries, and RH functions deployed in such systems map each query independently and can lead to system inefficiencies such as high cache-miss rates and page faults. The general problem of processing a burst of queries with a random underlying memory access pattern has received little attention. This problem affects the efficiency of gene sequence search systems which is the focus of this paper.

The problem of gene sequence search is identifying the species to which the given gene sequence belongs. Each species is represented by its genome: the complete set of DNA sequences  (typically represented as a sequence of "ACGT" characters) for that organism. The genome archives recording such sequences for multiple species are typically extensive. For instance, the European Bioinformatics Institute has reportedly stored around 160 Petabytes of raw DNA sequences as of 2018~\cite{EBI2019}. \textit{Bloom Filter}(BF), a probabilistic data structure used for membership testing, is the key component of state-of-the-art data structures such as \textit{Compact Bit-sliced Signature index} (COBS)~\cite{Bingmann2019COBSAC}, \textit{Repeated And Merged Bloom filters} (RAMBO)~\cite{gupta2021fast} and others~\cite{SSBT, Bingmann2019COBSAC, kmerBF} to efficiently search a given DNA sequence in petabytes of genomic data.
The search proceeds as follows. Each sequence of DNA is broken into kmers using a sliding window of size $k$ and stride 1 (see Figure \ref{fig:kmerBFInsert} for illustration), and each kmer is sequentially queried for membership inside the BF built over the DNA database. Thus, each sequence search leads to a burst of queries. RH functions, the primary reason for the success of BFs, being agnostic to the sequence of queries, distributes them randomly, causing us to fetch different caches/pages for each subsequent query. This is a highly inefficient usage of the cache /page, where only one bit is used from each fetch. Ideally, we want to co-locate the bits for subsequent queries to effectively utilize the cache and page mechanisms. 
\begin{figure}[t]
  \centering
  \includegraphics[scale=0.3]{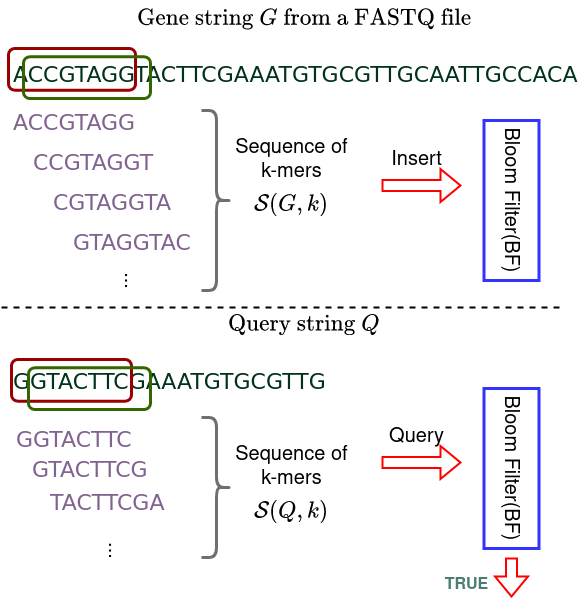}
  \caption{ An overview of gene string tokenization process followed by BF insertion and query. The long gene string of genome
  %from FASTQ file 
  is broken into kmers  (base substring of length $k$) using a moving window over the string, and then each kmer is inserted into the BF. While querying, again, the input sequence is broken into kmers  and the membership of each kmer is tested with the BF. If all the kmer pass the membership test, then the query is implied to be present in the corresponding 
  %FASTQ file.
  genome. }
   \vspace{-3mm}
  \label{fig:kmerBFInsert}
\end{figure}

One alternative might be to consider \textit{Locality-Sensitive Hash} (LSH) functions~\cite{IndykAndMotwani} 
that recognize the similarity between subsequent queries. However, the biggest strength of LSH in mapping similar tokens to the same value is a weakness in this scenario. As LSH maps similar keys to a single value, it does not preserve the identity of these keys. RH and LSH lie at two ends of a spectrum. LSH combines similar elements, causes them to have identical values, and induces a lapse of identity; RH, on the other hand, preserves identity but at the cost of being agnostic to any similarity. However, to improve gene-search system performance, we need a hash function that draws properties from both RH and LSH. In this paper, we provide a new hash family: Identity with locality (IDL) hash that offers a way to achieve the complete spectrum between RH and LSH. 

IDL hash family strikes a balance between RH and LSH families and provides the best of both worlds. It co-locates the values of similar input tokens but does not cause them to collide, thus preserving their identity (up to random collisions). Additionally, it randomly distributes dissimilar tokens across the entire range (see Figure \ref{fig:lphcollision}). IDL hash is a drop-in replacement of RH functions in the SOTA gene search methods. We generalize the BF, the core data structure in SOTA methods, by replacing traditional RH functions with \name functions. The \textit{Bloom filter} with \name hash function, called \namens-BF, is suitable for efficient sequential querying operations. \namens-BF makes the subsequent kmer in the DNA sequence co-locate in a range without colliding (up to random collisions). Thus, when a cache line  (alt. page) is fetched to access the bit of the current kmer, with high probability, the same cache line (alt. page) also contains the bit for subsequent kmer. Thus, \namens-BF reduces cache misses/page faults and improves the latency of indexing and querying DNA sequences. 

We describe the IDL hash family and provide a general recipe to construct \name functions in Section \ref{sec:family}. The design of \name function for gene search and its efficient implementation is described in Section \ref{sec:genesearch}. We analyze the false positive rate of \namens-BF in Section \ref{sec:analysis}. This analysis can be of independent theoretical interest. Also, we extensively evaluate IDL functions in SOTA methods of vanilla BF, COBS~\cite{Bingmann2019COBSAC}, and RAMBO~\cite{gupta2021fast} in experiments Section \ref{sec:exp}. 

Here, we assess the efficacy of IDL hash functions for gene search and find that they result in $\sim5\times$ fewer L1 and L3 cache misses during sequential queries compared to RH. The cache efficiency of IDL functions also leads to a faster performance in COBS, with a speedup of approximately $\sim1.4\times$ times for query and $1.45\times$ times for indexing. Furthermore, IDL functions demonstrate improvements in RAMBO, with up to $2.2\times$ times improvement in query and up to  $1.7\times$ improvement in indexing time.

\section{Background}
\begin{table}[t]
\centering
\caption{Notations used in the paper}
\label{tab:notations}
\resizebox{\linewidth}{!}{
\begin{tabular}{|l|l|l|l|}
\hline 
notation & description & notation & description\\
\hline 
$m$ & range/ Bloom filter size & $\mathcal{L}$ & Family of LSH\\
\hline 
$n$ & number of insertions & $\mathcal{H}$ &Family of universal hash\\
\hline 
$\eta$ & number of hash in BF & $\mathcal{I}$  & Family of \name hash\\
\hline 
$x, y$ & keys/kmers & $\rho(.)$& random function\\
\hline 
$q$ & query kmer& $\phi(.)$& LSH function \\
\hline 
$Q$ & query gene & $\psi(.)$& \name hash function\\
\hline 
$G$ & whole gene & $\mathcal{S}(.)$& function string $\rightarrow$ Set\\
\hline 
$k$ & kmer size& $\mathcal{M}(.)$& MinHash\\
\hline 
$t$ & sub-kmer size& $\zeta(.), \ \mathcal{J}(.)$& Jaccard Similarity\\
\hline 
$N$ & number of files & $p_1,p_2$ & probabilities\\
\hline 
$L$ & IDL random hash range & $r_1,r_2$ & distances\\
\hline 
 &  & $\mathbf{R}, \mathbf{N}$ & Real and Natural numbers \\
\hline
\end{tabular}}
   \vspace{-4mm}
\end{table}
\subsection{Random hash functions }
A hash function maps a given key from a set, say $U$, to integers in $[m] = \{0, 1, 2...{m-1}\}$. i.e. $\rho : U \rightarrow [m]$ for some $m \in \mathbf{N}$. A purely random hash function requires the storage of $O(|U|)$ and thus is not feasible to implement in practice. 
%Instead there are psuedo random hash function implementations such as murmurhash~\cite{murmur} and Xxhash~\cite{xxhash}. Additionally, there are families of hash functions that can be stored in the universality of murmur hash
However, there are families of hash functions that require O(1) storage and computation costs but only provide looser guarantees on the collisions of different keys. A family of the hash function $\mathcal{H}$ is called $t$-universal if  for all $x_1,x_2,\ldots, x_t \in U$ and $x_1\neq x_2 \ldots \neq x_t $, the following holds,
\begin{equation}
\mathbf{Pr}_{\rho \leftarrow \mathcal{H}}(\rho(x_1)= \rho(x_2)\ldots =\rho(x_t))\leq 1/m^{t-1}
\end{equation}
where $m$ is the range of the hash function $\rho$. i.e. the probability of $t$ keys colliding is bounded by $1/m^{t-1}$. The probability is over the choice of a hash function $\rho$ being chosen uniformly at random from the family $\mathcal{H}$. Generally, for most applications, 2-universal hash families suffice~\cite{mitzenmacher2008simple}. Popular implementations like \textit{MurmurHash}~\cite{murmur}, \textit{xxhash}~\cite{xxHash} are used in practice. We will refer to hash families promoting randomness in mapping as Random Hash (RH) families in the entire manuscript. However, it should be noted that in practice RH functions are generally implemented using $t$-universal hash families for small $t$ (eg. $t=2$ or $t=3$). Specifically, we use \textit{MurmurHash} for all our experiments.
\begin{figure}
    \centering
    \includegraphics[scale=0.3]{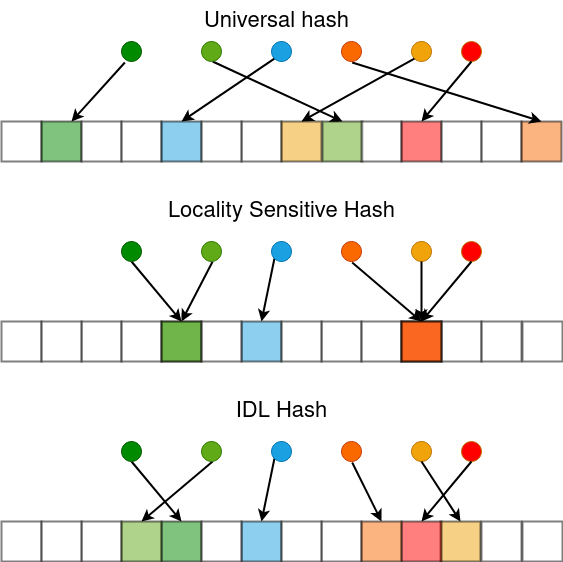}
    \caption{Illustration of different hash functions' behavior. While RH  disregards similarity in input space, LSH causes similar elements to collide. \namens, on the other hand maintains locality while discouraging collisions. }
      \vspace{-2.5mm}
    \label{fig:lphcollision}
\end{figure}

\begin{figure*}
  \centering
  \includegraphics[scale=0.3]{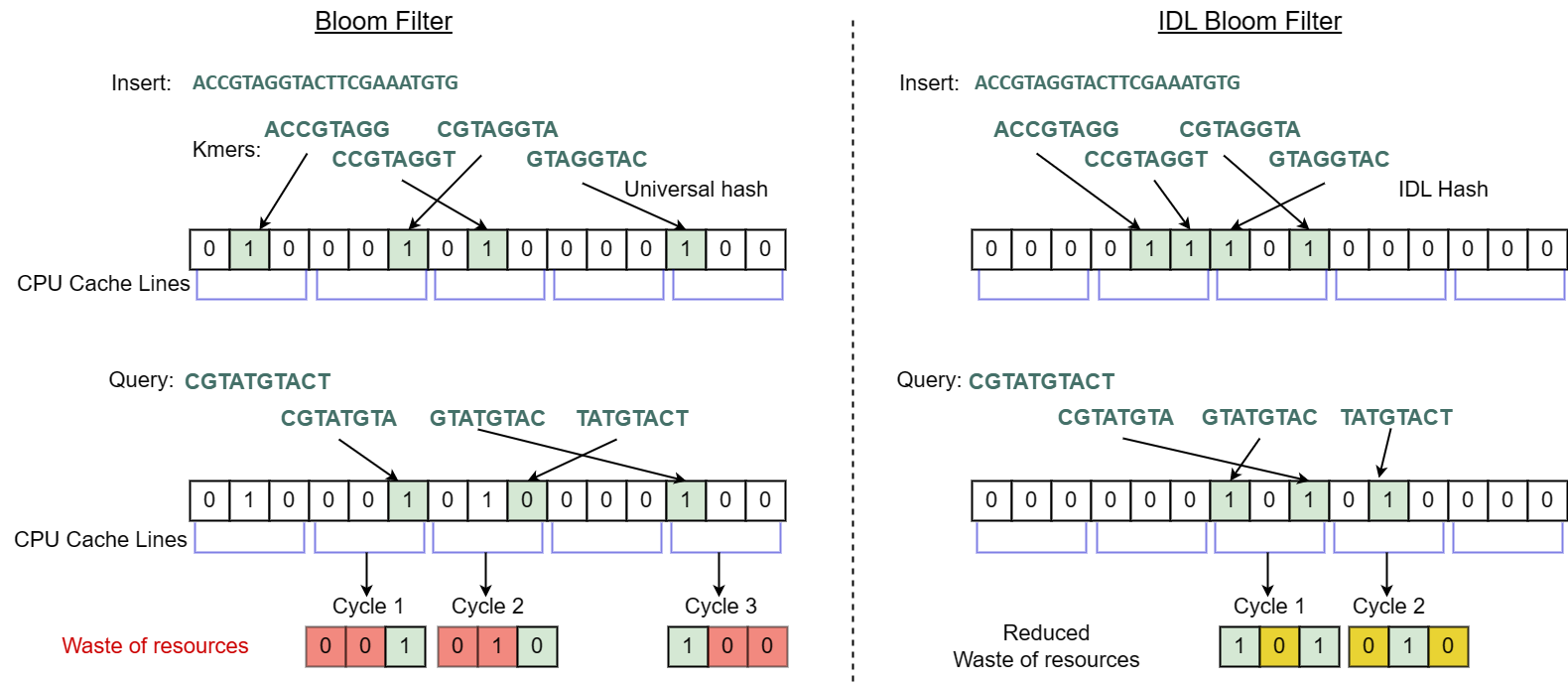}
  \caption{Illustration of gene sequence index and search using BF and IDL-BF. [Left] Traditional BFs causes inefficient usage of cache due to randomly mapping each subsequent kmer [Right] IDL-BF is cache-efficient, which uses the similarity of subsequent kmers to co-locate their bit signatures and thus use cache lines effectively.}
  \vspace{-2.5mm}
  \label{fig:ckbf}
\end{figure*}
\subsection{Locality sensitive hash}

\textit{Locality sensitive hashing }(LSH)~~\cite{IndykAndMotwani} is a technique from computational geometry that was introduced for approximate nearest neighbor search. An LSH family is defined as follows,

\begin{definition} [ LSH Family]\ A family  $\mathcal{L} = \{ \phi: U \rightarrow V \}$ is called
                $(r_1,r_2,p_a,p_b)$ sensitive with respect to a metric space $(U,d)$ if for any two points $x,y \in U$, $r_1 < r_2, p_a >p_b$ and $\phi$ chosen uniformly from $\mathcal{L}$ satisfies the following:
               
                \begin{itemize}
                                \item if $d(x,y)\le r_1$ then ${Pr}_{\phi \leftarrow \mathcal{L}}(\phi(x) = \phi(y)) \ge p_a$
                                \item if $ d(x,y)\ge r_2$ then ${Pr}_{\phi \leftarrow \mathcal{L}}(\phi(x) = \phi(y)) \le p_b$
                \end{itemize}
                \label{def:lsh}
\end{definition}
This implies that if the given two points $x$ and $y$ are close w.r.t distance metric $d$, the corresponding hash values $\phi(x)$ and $\phi(y)$ are the same (i.e., they collide) with high probability, and points that are distinct w.r.t $d$ have a low probability of collision. Similarity metrics (defined as the inverse of distance metric) are sometimes used to interpret LSH. Many similarity metrics accept the LSH family of functions. Some examples are \textit{Jaccard}~\cite{broder1997resemblance}, \textit{Euclidean}~\cite{datar2004locality} and \textit{angular distances}~\cite{charikar2002similarity}.

While $\phi \in \mathcal{L}, \phi : U \rightarrow V$ outputs a value in $V$, for some applications, we want the functions to map to a range $[m]$ for some $m \in \mathbf{N}$. In such a case, we use the rehashing trick and rehash the output of the LSH function into the range $[m]$. Let $\phi: U \rightarrow V$ be an LSH function with the probability of collision $\zeta(x,y)$ for $x,y \in U$. Then the probability of collision of LSH function $\phi_m$ that is rehashed to range $m$ is 
\begin{equation}
    \mathbf{Pr}( \phi_m(x) = \phi_m(y)) = \zeta(x,y) + (1 - \zeta(x,y)) \frac{1}{m}
\end{equation}

\paragraph{\textbf{MinHash}}
Min-wise independent permutations~\cite{broder1997resemblance}, or \textit{MinHash}, is an LSH function used to estimate the similarity between sets. Let $U$ be a set and $X, Y \subseteq U$. Let $\rho$ define a permutation over $U$. Then \textit{MinHash} $\mathcal{M}$ is defined as
\begin{equation}
    \mathcal{M}(X) = \min \{ \rho(x) | x \in X \}
\end{equation}
The probability of collision for $X, Y \subseteq U$ is given by \textit{Jaccard similarity} ($\mathcal{J}$) between the two sets.
\begin{equation}
\mathbf{Pr}(\mathcal{M}(x) = \mathcal{M}(y)) = \mathcal{J}(X, Y) = \frac{|X\cap Y|}{|X\cup Y|}
\end{equation}
\subsection{Bloom filter}
\textit{Bloom Filter} (BF) is a probabilistic data structure used for membership testing. It is parameterized by a range $m$ and a number of hash functions $\eta$. Given a set $A$, each element of $A$, say $a$, is hashed using $\eta$ hash functions $\rho_1(a)$, $\rho_2(a)$,$\ldots$, $\rho_{\eta}(a)$ of range $m$ and the bits at these locations are set to $1$.  To answer the query "$q\in A?$", the same hash functions $\rho_1(q)$, $\rho_2(q)$,$\ldots$, $\rho_{\eta}(q)$ are used to get bit locations for the item $q$. If all the bits are $1$ at these hash locations, $q$ is inferred as part of the set $\textit{A}$.

One important metric evaluating the quality of BF is False Positive Rate (FPR). It is the rate at which elements that do not belong to the set $A$ are reported as members of $A$. Under the assumption of complete independence (i.e., hash functions are purely random), the FPR, say $\epsilon$, for BF of range $m$, number of hash functions $\eta$ and populated with a set $A$ of size $|A|=n$ is given by
\begin{equation}
    \epsilon = \left(1-\left[1-{\frac {1}{m}}\right]^{\eta n}\right)^{\eta} \approx \left(1-e^{-\eta n/m}\right)^{\eta} 
\end{equation}
For a given $m$, the FPR is minimized when $\eta = (\textrm{ln}(2)m)/n $. Thus, under optimal choice of $\eta$, one can determine the range $m$ for a required FPR $\epsilon$, $m = - (n \textrm{ln} (\epsilon))/(\textrm{ln}^2 (2))$. Note that BF has no false negative rate as any seen item $x$ will set all the bits $\rho_1(x), \rho_2(x),.. \rho_{\eta}(x)$.
\subsection{Gene search with vanilla BF} \label{subsec:genesearch}
In this section, we will explain how gene search works with vanilla BF. Given a DNA string for query, $Q$, the task is to determine if the query $Q$ belongs to given species represented by its full-length genome $G$. The gene string of the full-length genome $G$ is usually very large, with an average length of $100$ Million~\cite{gupta2021fast}. Given this large genome length, it is practically infeasible to match the string precisely. Additionally, it is useful to understand how much of the given query gene string $Q$, if not all, matches with a given genome $G$. The standard practice to address this is to represent the genome $G$ as a set of $k$ length substrings called kmers~\cite{kmerSelection}.

Given a gene string, say $G$ ( and $Q$), we divide this string into a sequence of substrings of $k$ base pairs called kmers using a moving window. This is illustrated in Figure \ref{fig:ckbf}. Let this sequence be generated by function $\mathcal{S}(.)$, 
\begin{equation}
    \mathcal{S}(G, k) = \{G[i:i+k]\}_{i=1}^{(|G| - k)}
\end{equation}
% \begin{equation}
%     \mathcal{S}(Q, k) = \{Q[i:i+k]\}_{i=1}^{(|Q| - k)}
% \end{equation}
The Membership Testing (MT) and Multiple Set Membership Testing (MSMT) of a gene string query $Q$ is defined as follows,
\begin{definition} Membership Testing(MT) of given query gene string $Q$ in genome of species $G$ is defined as,
\begin{equation}
    \mathbf{MT}(Q, G) =1 \ \text{iff} \ \mathcal{S}(Q, k) \subseteq \mathcal{S}(G, k)
\end{equation}
\end{definition}

\begin{definition} The Multiple Set Membership Testing(MSMT) of a given query gene string $Q$ in a set of genomes of species $\{G_i\}_{i=1}^N$ is defined as,

\begin{equation}
\mathbf{MSMT}(Q, \{G_i\}_{i=1}^N) = \{ \mathbf{MT}(Q, G_i)\}_{i=1}^N 
\end{equation}
\end{definition}
MT or MSMT problems can be approximately solved using BF. In fact, BF and its variants are the state-of-the-art solutions to membership testing problems in gene search~\cite{bigSI, elworth2020petabytes}.

The FPR of the query $\mathcal{S}(Q, k)$ on the BF indexed with $\mathcal{S}(G, k)$ is given by 
\begin{equation}
    \epsilon_{MT} = \mathbf{Pr}(\ BF(Q)=1\ |\ \mathcal{S}(Q,k) \not \subset \mathcal{S}(G,k)\ )
\end{equation}
% is calculated empirically as follows. 
% For a given set of queries as $\mathcal{Q}$, the BF reports the set $\mathcal{Q}'$, $ \mathcal{Q}\subseteq \mathcal{Q}'$ as the set of positives, then the false positive rate is calculated as
% $$
% \text{FPR} = \frac{|\mathcal{Q}'/\ \mathcal{Q}|}{|\mathcal{Q}'|}
% $$

% $$
% \text{FPR} = \frac{|\{q | q \in \mathcal{Q}', q \notin \mathcal D\}|}{|\mathcal Q|},
% $$
% where $|\mathcal Q|$ is the cardinality of the set $\mathcal Q$. 

For multiple-set membership testing on $N$ genomes $\{G_i\}_{i=1}^N$,the FPR is defined as follows,

\begin{equation}
    \epsilon_{MSMT} = \frac{1}{N}\sum_i^N \mathbf{Pr}(\ BF_i(Q)=1\ |\ \mathcal{S}(Q,k) \not \subset \mathcal{S}(G_i,k)\ )
\end{equation}

In practice, the probability $Pr(.)$ is approximated by the empirical average over a large number of queries.
% Assume we are interested in testing membership of a set of queries $\mathcal Q$ in multiple BF representing $\mathcal D_1, \dots, \mathcal D_N$, for $N$ BF. If the algorithm reports set $\mathcal{Q}'$, $ \mathcal{Q}\subseteq \mathcal{Q}'$ as the positives on BF $\mathcal D_i$, where $i \in [N]$, then the false positive rate is

% $$
% \text{FPR}_\text{MSMT} = \frac{\sum_i|\{q | q \in \mathcal{Q}_i', q \notin \mathcal D_i\}|}{N|\mathcal Q|}.
% $$

% In this paper we improve the efficiency of BF based solutions for gene search by introducing new class of hash functions - \name.

 \label{sec:background}
\section{Related Works}
\begin{figure*}
    \centering
    \includegraphics[scale=0.3]{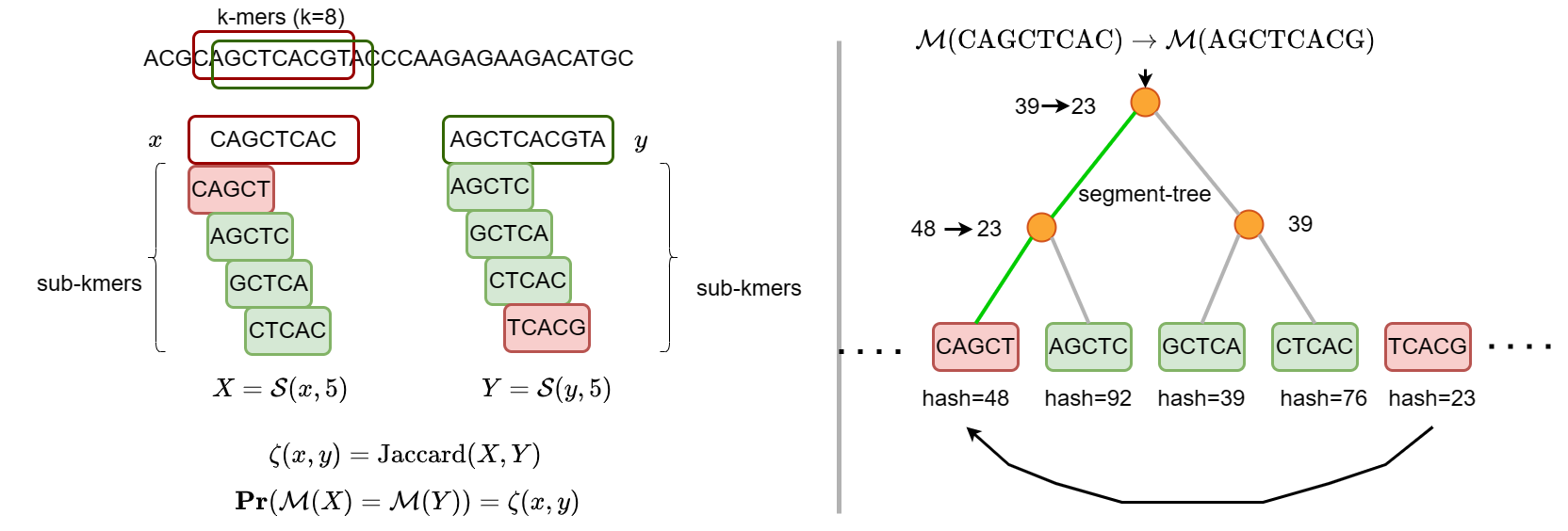}
    \caption{[Left: LSH computation on kmer] Each kmer (k=8 in this case) is split into sub-kmers of length t (t=5) and then Min-hash is applied to this set of sub-kmers. The probability of collision of two kmers is equal to the \textit{Jaccard similarity} between the two sets. [Right: Rolling min-hash] As two subsequent kmers only differ in one sub-kmer, we can reuse the hash computations from the previous kmer for the current kmer. We build a complete segment tree data-structure on the sub-kmer hashes for the first kmer. For the subsequent kmers we need to replace exactly one existing leaf-node corresponding to sub-kmer that is not present in the current kmer and replace it with the new incoming sub-kmer. Thus, for each kmer (after first one), we only need to compute one hash of sub-kmer and update $\textrm{log}(k-t)$ min values in the tree.}
     \vspace{-2.5mm}
     \label{fig:mhdetails}
\end{figure*}

\subsection{Locality preserving hashing}
The idea of our proposed IDL hash function family is closely related to a class of hash functions called \textit{Locality Preserving Hash} (LPH) \cite{indyk1997locality} and \textit{non-expansive} (NE) \cite{linial1996non} hash functions. NE hash functions are defined from a set $[U]$ to a set $[R]$  ($U >> R$) where $[m] = \{0, 1, ..., m-1\}$. The property of this hash function family is that any subset, say S,  of [U] of size $|S| < C\sqrt{R}$ can be mapped without any collisions with probability $\geq \frac{1}{2}$ and all of the functions of this family are non-expansive, i.e., $d(f(p), f(q)) \leq d(p,q)$. This was the first demonstration that hashing could be non-expansive, even with a small range of $R$. The idea was extended to multi-dimensional cube input spaces $\{-U, ... U\}^d$ in \cite{indyk1997locality}, and the hash functions were termed LPH. The underlying idea of LPH hash functions is still that of using NE hash functions. 

While the idea of NE and LPH hash functions are similar to the goal of IDL hash functions, NE and LPH hash functions are defined on integer domains and can be potentially extended to $R^d$ domains. However, it is unclear if we can extend such hash functions to string domains or, more generally, to sets. On the other hand, existing LSH functions are already non-expansive for points in input space which are very close with a high probability. As we will show, we can construct IDL hash function families that are locality-sensitive and identity-preserving using existing LSH functions. Thus, we can create an IDL hash function family for any domain for which LSH functions have been developed. Therefore IDL hash function families are widely applicable. 

\subsection{State-of-the-art in gene sequence search systems}
The gene sequence search problem has been formulated as a document retrieval task, wherein a complete genome sequence is treated as a document, and a DNA of any arbitrary length is queried to retrieve the document(s) containing it. Initially, BLAST \cite{altschul1990basic} was used to address this problem for small-scale data. However, with the exponential growth of gene archives \cite{bigSI,schatz2013dna,ENA2022}, database practitioners have shifted towards using inexpensive and scalable indexes that grow linearly with the number of files. The \textit{Sequence Bloom Trees} (SBTs) \cite{solomon2017improved,crainiceanu2013bloofi} and BIGSI/COBS \cite{bigSI,Bingmann2019COBSAC} were among the first approaches to use BFs, owing to their small size ($2.5$ to $7$ bits per key). While SBTs theoretically offer sub-linear query times, they are larger in size and slower than BIGSI \cite{bigSI} and its successor COBS (a linear query time index) on real datasets \cite{Bingmann2019COBSAC}.
BIGSI and COBS maintain one BF for each file. Consequently, the query time grows linearly in the number of files. 
Recently, a BF-based index, RAMBO \cite{gupta2021fast}, has overcome this limitation by randomly merging and repeating BFs, thereby providing sub-linear query time with linear index size.

All of these advanced searches for gene sequences are constructed using a BF that utilizes the RH functions. As we will see, the system performance of such search systems can be improved by simply replacing RH functions with IDL functions. 

\subsection{Cache-efficient Bloom filters}
\textit{Blocked Bloom filter} (BBF) \cite{putze2007cache} improved the BF by reducing the cache misses caused by the $\eta$ RH functions. It divides the BF into blocks of cache-line size in the BBF. As a result, all $\eta$ hash functions will hash to the same block for any given query. BBF can reduce up to $\eta-1$ additional cache misses at the expense of an increase in False positive rates. On the other hand, \namens-BF takes a different approach to reduce cache misses of subsequent queries by co-locating similar kmers. It can reduce cache misses by up to $O(k)$ times. Generally $k=31 >> \eta$. Clearly, BBF and IDL-BF are orthogonal approaches that can easily be integrated with each other.

% previous version
% The cache efficiency of the Bloom filter is explored in the Blocked Bloom filter \cite{putze2007cache}. It addresses the occurrence of cache misses due to the $\eta$ hash functions. The Blocked Bloom filter splits the BF into cache-line size blocks. Subsequently, all the $\eta$ hash functions hash to the same block for a given query. Blocked BF reduces at most $\eta-1$ additional cache misses with a small gain in False positive rates. \name-BF's alternate approach of reducing cache misses harnesses the inherent data (DNA) structure. It can reduce at most $(k-t+1)$ cache misses. Existing cache efficiency Bloom filters approached can be easily complemented with \name-BF. The use of LSH functions has been limited in the Bloom filter. The distance-sensitive Bloom filter uses LSH, which can be used for searching similar elements \cite{engels2021practical}. Additionally, and specifically, the use of the Bloom filter for genomics has accelerated \cite{elworth2020petabytes}. In \cite{kmerBF}, authors worked on improving the BF for kmers by using the limitedness of base pairs (only 4) to reduce the FPR by doing additional queries and hence increasing the query time. 

% he work by Chikhi et.al. \cite{kmerSelection} confirmed k=31/51/71 to be optimal for $k$-mer set size and uniqueness. For the ENA dataset, most of the best methods \cite{bigSI} \cite{Bingmann2019COBSAC} \cite{SBT} \cite{HowDeSBT} use k=31, partially also because it is small enough to be represented as a 64-bit integer variable with 2-bit encoding.

 \label{sec:related}

\section{{\longname } family} \label{sec:family}
\label{IDLfamily}

In this section, we define the \textit{\longname}~(\name hash) family and show how to construct an instance of this family when the co-domain is a range of integers $[m] = \{0,\ldots,m-1\}$. 
\begin{definition}(\longname)
A family $\mathcal{I} = \{ \psi: U \rightarrow V\}$ is $(r_1, r_2, p_1, p_2)$ sensitive and $(L)$ preserving with respect to the metric spaces $(U,d_U)$ and $(V, d_V)$ if $r_1 < r_2$, $p_1 > p_2$, $L>0$ and for any $x,y \in U$, the following holds,
\begin{itemize}
    \item if $d_U(x, y) \leq r_1$, then 
    $$\mathbf{Pr}_{\psi \leftarrow \mathcal{I}} ((\psi(x) \neq \psi(y)) \wedge d_V(\psi(x), \psi(y)) < L) \geq p_1$$
    
    \item if $d_U(x, y) > r_2$, then $\mathbf{Pr}_{\psi \leftarrow \mathcal{I}} (d_V(\psi(x), \psi(y)) < L) \leq p_2$
\end{itemize}
\end{definition}

%\gaurav{suggestion for a third statement: $if d_U(x, y) \leq r_1$, then $\mathbf{Pr}_{h \leftarrow \mathcal{H}} ((h(x) = h(y)) \textrm{ and } d_V(x, y) < L) \leq 1/L$}

The \name function causes the close-by elements in the input metric space to lie inside a ball of diameter $L$ in the target metric space without colliding. In the context of using hashing for indexing or sketching, we want $V$ to be a range $[m] = \{0,...,m-1\}, m \in \mathbf{N}$ with $d_V(a,b)$ defined as $d_V(a,b) = |a-b|$.

Following is a generic way to construct an \name function from a given LSH function on a metric space $U$. Consider an LSH function family $\mathcal{L}$ defined on a metric space $(U, d_U)$. It is often the case that $V \neq [m]$. We give a construction for this case which is also applicable when $V=[m]$. Consider two random hash function families $\mathcal{R}_1: V \rightarrow [m]$ and  $\mathcal{R}_2: U \rightarrow [L]$. Then we define a hash function family $\mathcal{I} : \{\psi: U \rightarrow [m+L]\}$
\begin{equation}
\mathcal{I}: \{\psi(x) = \rho_1(\phi(x)) + \rho_2(x) | \; \phi \in \mathcal{L}, \rho_1 \in \mathcal{R}_1, \rho_2 \in \mathcal{R}_2 \}
\end{equation}
The $\mathcal{I}$ is an \name family, and we state this result in the theorem below, 
\begin{algorithm}[t]
% \setstretch{1.14}
\begin{algorithmic}
\STATE {\bf IDL-BF INSERT}
% $\psi: U\rightarrow m+L$: IDL hash function
% $\rho :U \rightarrow [L]$ 
\STATE {\bf Input:} $G$: genome string, $\mathcal{M}: MinHash$, $\rho \in \mathcal{R}$: random function, $k$: kmer size, $t$: sub-kmer size, $\mathbb{B}$: IDL-BF of size $m$
\STATE $\{x_i\}_{i=1}^{|G|-k+1}= \mathcal{S}(G, k)$\\
\STATE Init  IDL-BF : $ \mathbb{B} = \mathbf{0}[m]$\\
\FOR{ $i \in 1..|G|-k+1$}
    \FOR{ $j \in 1..\eta$}
        \STATE $loc_j \leftarrow \psi(x_i, t, L) = \mathcal{M}(x_i, t)+ \rho(x_i)$ : seed=  $j$ \\
        \STATE $\mathbb{B}[loc_j] = \mathbf{1}$\\
    \ENDFOR
\ENDFOR
\end{algorithmic}
  \caption{Insert for IDL-BF}
  \label{alg:IDLBFInsert}
\end{algorithm}
\begin{theorem} (general \name construction) \label{thm:genconstruction}
Let $\phi$ be drawn from a $(r_1, r_2, p_1, p_2)$ sensitive LSH family say $\mathcal{L}$ and $\rho_1, \rho_2$ be drawn from a random hash family, say $\mathcal{R}_1 : V \rightarrow [m], \mathcal{R}_2: U \rightarrow [L]$ respectively. Then the family of hash functions defined by 
\begin{equation}
    \mathcal{I} : \{\psi(x) = \rho_1(\phi(x)) + \rho_2(x) \}
\end{equation}
is a $(r_1, r_2, \frac{L-1}{L}p_1, \frac{L}{m} + p_2)$ sensitive and $L$ preserving IDL family.
\end{theorem}
\begin{proof}
The proof is deferred to additional materials.
\end{proof}
The construction first uses a given LSH function on $x$ to find a location $\rho_1(l(x))$ in $[m]$. It then applies a universal hash function to the key $x$ to distribute it locally near this location in the range $\{\rho_1(l(x)), \  \rho_1(l(x))+1, \  ... , \ \rho_1(l(x))+ L-1\}$

\section{Gene sequence search with \name hash functions}
\label{sec:genesearch}
We propose to change the RH family used in gene sequence search systems with the \name family of functions. The goal is to co-locate the bit of subsequent kmers so that we can effectively use the memory hierarchy while accessing the BF memory in these data-structures for sequential probing. Specifically, with high probability, the bit of the next kmer should be found in the same cache line (alt. page when accessing disk) as that of the current kmer avoiding additional access to RAM or disk. 

\subsection{\name function for gene sequence search}

First, we show how to construct an \name function for our purposes. We will use the general construction for \name functions as demonstrated in Theorem~\ref{thm:genconstruction}. We need to define an LSH function on the kmers, ensuring that subsequent kmers have a high probability of collision. The subsequent kmers have significant overlap in their representation, and this can be used to design a similarity function.

Consider two kmers $x, y$ of length $k$. We create set of sub-kmers, $X, Y$, of size $t$, where $X = \mathcal{S}(x, t)$ and $Y = \mathcal{S}(y, t)$.

For the similarity between the two kmers, $\zeta(x,y)$, we consider the \textit{Jaccard similarity} ($\mathcal{J}$) of the two sets of sub-kmers,
\begin{equation}
    \zeta(x, y) = \mathcal{J}(X, Y)
\end{equation}
The similarity metric accepts the LSH function of \textit{MinHash} ($\mathcal{M}$).

The LSH function $\phi(.)$ then, is
\begin{equation}
    \phi(x) = \mathcal{M}(\mathcal{S}(x, t))
\end{equation}
where $t$ is the sub-kmer size. An illustration of how a kmer is broken into a set of sub-kmers is presented in Figure \ref{fig:mhdetails}. The probability of collision in the case of using \textit{MinHash} is also equal to the $\zeta(x,y) = \mathcal{J}(X, Y)$

\textbf{Choice of $t$:} The choice of $t$ will determine how the \name function behaves for gene search. For instance, setting $t=k$ essentially causes \name to behave like an RH function. Thus, $t=k$ ignores similarity among kmers. Setting $t=1$ gives us the maximum similarity with subsequent kmers. However, it restricts the hash space of the \textit{MinHash} function as we only have four distinct base pairs. Thus, it is crucial that $t$ cannot be too large or too small. We find that $t=12$ to $t=20$ work well in practice for a kmer size $k=31$ which is the standard.
\begin{algorithm}[t]
% \setstretch{1.14}
\begin{algorithmic}
\STATE {\bf IDL-BF QUERY}
\STATE {\bf Input:} $Q$: query string, $\mathcal{M}: MinHash$, $\rho \in \mathcal{R}$: random hash, $k$: kmer size, $t$: sub-kmer size, $\mathbb{B}$: Indexed IDL-BF 
\STATE $\{q_i\}_{i=1}^{|Q|-k+1}= \mathcal{S}(Q, k)$\\
\FOR{ $i \in 1..|Q|-k+1$}
    \FOR{ $j \in 1..\eta$}
        \STATE Init:  $bit = 0$, Membership
        \STATE $loc_j \leftarrow \psi(q_i, t, L) = \mathcal{M}(q_i, t)+ \rho(q_i)$ : seed=$j$ \\
        \STATE $bit = bit \And \mathbb{B}[loc_j]$\\
        \IF{$bit == 0$}
            \STATE Membership = False\\
            \STATE EXIT\\
        \ENDIF
    \ENDFOR
\ENDFOR
\STATE Membership = True
 \end{algorithmic}
  \caption{Query from IDL-BF}
  \label{alg:IDLBFQuery}
\end{algorithm}

We choose the random hash functions $\rho_1(.)$ and $\rho_2(.)$ as suggested in the Theorem~\ref{thm:genconstruction}. Another hyperparameter for the function $\rho_2$ and the \name function itself is the locality parameter $L$.

\textbf{Choice of $L$:} A critical choice while choosing $\rho_2$ is that of $L$. It determines how close the subsequent kmers will be hashed. While we want the subsequent kmers to be close, we also want to discourage their collisions. In our ablation study, we find that using page size for $L$ works well both in terms of promoting locality and discouraging collisions on RAM. 

This completes the construction of an \name function in accordance with Theorem~\ref{thm:genconstruction}.

\subsection{SOTA gene search data structures with IDL}

In order to obtain the IDL versions of SOTA gene search data structures, we simply need to substitute the RH functions with the IDL function mentioned in the preceding section. In this paper, we consider three data structures for gene search. 

First is the vanilla BF which is described in Section~\ref{sec:background}. We call BF using IDL hash functions as IDL-BF.  Figure \ref{fig:ckbf} provides a comparative illustration of the operation of \namens-BF and BF. The complete pseudocode of insertion and querying with \namens-BF for gene search is described in Algorithms \ref{alg:IDLBFInsert} and \ref{alg:IDLBFQuery}. As with the standard BF, we must select the parameters of range, $m$, and a number of independent hash functions, $\eta$, with \namens-BF. The $m$ is largely determined by available resources, as with standard BF. In our experiments, we find that $\eta$, as suggested for standard BF, also works well with \namens-BF. 

Secondly, we consider \textit{Compact Bit-sliced Signature index} (COBS) \cite{Bingmann2019COBSAC}, which is essentially an array of BFs created for multiple genome files. 
To create the IDL version of COBS, we substitute the RH in each BF with the IDL hash function, resulting in IDL-COBS. The MSMT operation on COBS consists of independent MT on each BF. Therefore, naively replacing RH with \name and using a standard set of BF parameters in each BF leads to IDL-COBS.

% The various parameters for COBS are chosen in the fashion similar to IDL-BF for each of the genomes.

Thirdly, we consider \textit{Repeated And Merged Bloom Filter} (RAMBO). Unlike COBS, RAMBO creates multiple BFs, with each BF representing a subset of randomly selected files. The parameters for RAMBO include the number of BFs in each repetition ($B$) and the number of repetitions ($R$). For $N$ files, $B$ is typically set to $O(\sqrt{N})$ and $R$ to $O(\log N)$. RAMBO's parameters are independent of the MT function, and therefore, the replacement of RH with \name in each BF of RAMBO is compatible with the standard RAMBO parameters. The modified RAMBO index obtained by using the \name hash function is referred to as IDL-RAMBO.
% We replace RH function in each of these BFs with the IDL hash function and call the resulting data structure as IDL-RAMBO.

% \textbf{Array of \name-Bloom Filters}
% A multiple-set membership testing is needed for searching queries to return matching species. For this, we extend \namens-BF to an Array of \name-Bloom Filters (ABF). Like ABF, \name-ABF achieves multiple set membership testing by indexing a distinct \namens-BF for each set. During querying, the query is tested against each \namens-BF for membership in the corresponding set. The motivation and working of the Array of BF are noted in algorithm \ref{alg:ABFInsertQuery} in the appendix.

\subsection{Efficient generation of IDL hash for subsequent kmers}
A naive implementation of \name hash for gene search would increase the hashing costs as compared to RH functions used in BF. This can potentially outweigh the benefits obtained due to cache-locality.
In this section, we outline how to efficiently implement the \name function using \textit{MinHash}. 
\begin{algorithm}[t]
% \setstretch{1.14}
\begin{algorithmic}
% \SetAlgoLined
\STATE {\bf Input:} Q :query or genome string, $k$: k-mer size, $t$: sub-k-mer size, $\rho \in \mathcal{R}:U \rightarrow [N]$. where U is the set of all k-mers and N is a large integer  \\

\STATE $\{x_i\}_{i=1}^{|Q|-k+1} = \mathcal{S}(Q, k)$  \ \ \ \ \#set of kmers  \\
\STATE $\{y_j\}_{j=1}^{k-t+1} = \mathcal{S}(x_0, t)$  \ \ \ \ \#set of sub-kmers \\
\STATE $\mathbf{T} \leftarrow \textrm{segment-tree}\left(\{\rho(y_j)\}_{j=1}^{k-t+1}\right)$\\
\STATE $idx \leftarrow 0$\\

\FOR{$x \in \{x_1, x_2, ... x_{|Q|-k+1}\}$ }
    \STATE $y_{in} = \mathcal{S}(x, t)[-1]$  \ \ \ \ \#incoming sub-kmer\\
    \STATE $y_{out} = T[idx]$  \ \ \ \ \#outgoing sub-kmer \\
    \STATE $\textrm{update-tree}(T, y_{out}, y_{in})$
    \STATE yield $T.\textrm{rootvalue}$\\
    \STATE $idx = (idx+1)\%(k-t+1)$\\
\ENDFOR
 \end{algorithmic}
  \caption{Rolling-min-hash over k-mers of gene string using segment tree when tree is complete ($(k{-}t{+}1)$ is power of 2.)}
  \label{alg:rollingmh}
\end{algorithm}

\subsubsection{\textit{MinHash} $(\mathcal{M})$ computation using random hash function}
Consider a set $A \subseteq[ N]$ where $N \in \mathbf{N}$. The exact \textit{MinHash} computation of $A$ is performed by using a permutation of $[N]$, say $\textrm{perm}$, and then computing
\begin{equation}
    \mathcal{M}(A) = \textrm{min}_{a \in A} (\textrm{perm}(a))
\end{equation}

However, $N$ is usually quite large, and computing and storing the permutation of $[N]$ is expensive. Instead, we use RH functions to compute \textit{MinHash} as follows,
\begin{equation}
    \mathcal{M}(A) = \textrm{min}_{a \in A} (\rho(a))
\end{equation}
where $\rho$ is drawn from an RH family. For RH functions, we use implementations such as \textit{MurmurHash}.

\subsubsection{Rolling \textit{MinHash} computation for subsequent kmers}
In order to see an overall improvement in indexing and query times, it is important to control the cost of hashing. If the \textit{MinHash} based \namens-BF is naively implemented, we would require computing $(k{-}t{+}1)$ random hash computations, $(k{-}t)$ comparisons for \textit{MinHash}, and one additional random hash with range $L$. This is expensive compared to a single hash computation for traditional BF. However, this can be reduced to just one additional random hash computation and $\log(k{-}t)$ comparisons.

% situation can be remedied using clever data structures.

% \paragraph{}
We observe that the sequence of sub-kmers $\mathcal{S}(x, t)$ and $\mathcal{S}(y, t)$ for subsequent kmers $x$ and $y$,  have exactly one sub-kmer different, namely $\mathcal{S}(x)[0]$ (outgoing sub-kmer) and $\mathcal{S}(y)[-1]$ (incoming sub-kmer). Thus, we can reuse the hash computations for other sub-kmers. We keep these hash values in a segment tree data structure, with each internal node maintaining the min of the sub-tree. For subsequent kmer, an incoming sub-kmer replaces the outgoing sub-kmer in the tree, and the segment tree updates and ensures the minimum of all leaves is kept at the root.
% we only need to update the path from the replaced leaf node to the root to compute the minimum of all sub-kmers. 
This process is illustrated in Figure \ref{fig:mhdetails}. Thus for every kmer (except the first), we need to perform one hash and $log(t-k)$ comparisons for computing \textit{MinHash}. Refer to the rolling \textit{MinHash}'s Algorithm \ref{alg:rollingmh} for details.

\subsubsection{Densified one-permutation hashing}
With rolling \textit{MinHash} computation for an \namens-BF with $\eta$ independent hash functions and $m$ range, for each kmer, we need to compute 2$\eta$ hash values (one for \textit{MinHash} and one for local hashing in $[L]$). In a standard BF with RH functions, we need to compute $\eta$ hashes. We can further reduce the cost of hashing by using \textit{densified one-permutation hashing} \cite{shrivastava2014densifying} for \textit{MinHash} computations. This enables us to compute multiple \textit{MinHash} at the cost of single hash computation. Thus, we need to only compute $(\eta + 1)$ hash values for every kmer (except for the first kmer, which is used to populate the segment tree in rolling \textit{MinHash} computation.).

%\gaurav{can uncomment if we have space}
% \input{tables/hashes} The final number of hash computations is noted in table \ref{tab:hashes}

\section{Analysis} \label{sec:analysis}
This section provides a generic analysis of \namens-BF which is at the core of efficent SOTA gene search data structures such as COBS and RAMBO. We make some assumptions about the data inserted into the BF and the relation between this data and the query. These assumptions are motivated by applications such as gene search.
\begin{table}[t]
\centering
\caption{Evaluation of assumption 1 on gene-search data. For $w_1{=}k{=}31$, we estimate the probability of faraway tokens having 0 \textit{Jaccard similarity}. As we can the probability for this is very close to 1.0 across different sized genome files. }
\label{tab:assumption1}
\resizebox{\linewidth}{!}{
\begin{tabular}{|l|c|c|}
\hline
\textbf{File name}       & \textbf{Size}  & $\mathbf{Pr}(\mathbf{1}(\mathcal{J}(\mathcal{S}(x_i, t), \mathcal{S}(x_j, t)) = 0)), |i-j| \geq k$ \\ \hline
ERR105106  & 384MB & 0.99944                                                                                            \\ \hline
ERR1102508 & 680MB & 0.99999                                                                                            \\ \hline
SRR1816544 & 1.3GB & 0.99999                                                                                            \\ \hline
SRR1955610 & 3.8GB & 1.0                                                                                                \\ \hline
\end{tabular}}
\vspace{-2mm}
\end{table}
\begin{figure*}
  \centering
  \includegraphics[width=\textwidth]{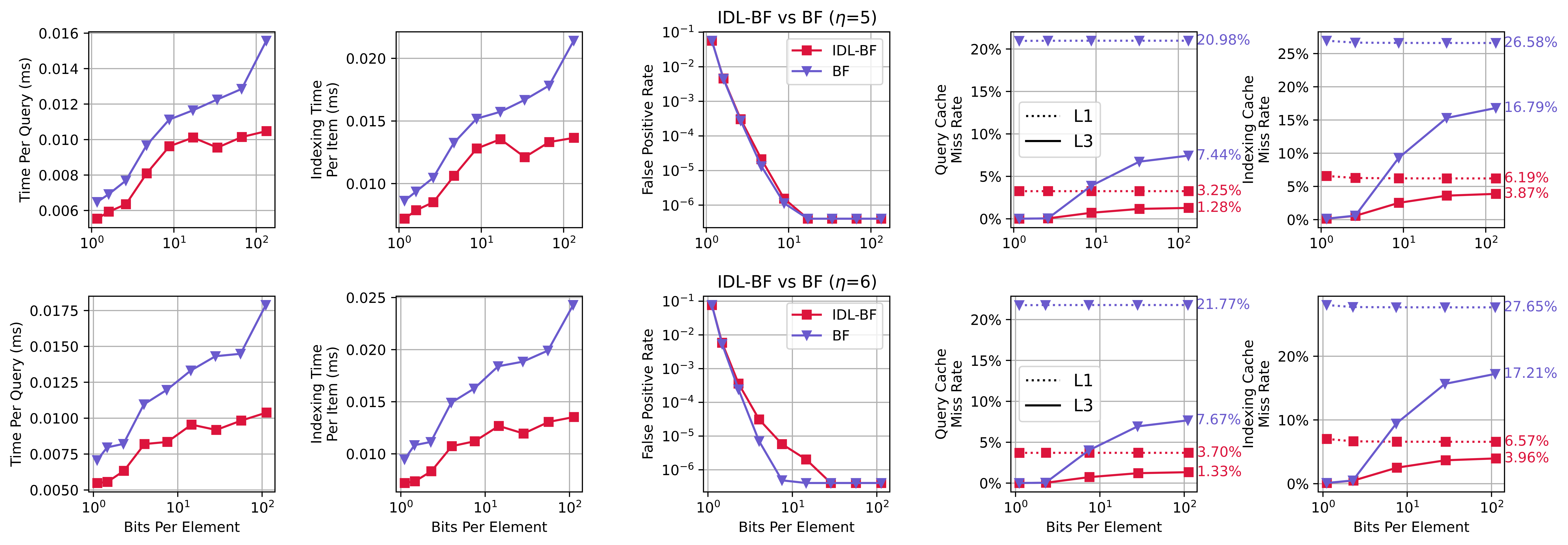}
  \caption{The impact of different sizes of \namens-BF vs BF on query time, indexing time, FPR, and cache miss rates, averaged over 5 runs. Query Time and indexing time of \namens-BF grow significantly slower when compared with BF, while achieving similar FPR. For the same BF size, \namens-BF achieves up to 41.9\% and 44.3\% reduction in query time and indexing time, respectively. The reductions in query and indexing time are accounted for by the reductions in cache miss rate achieved by \namens-BF. For the same BF size, \namens-BF achieves up to 76.2\% and 77.0\% reduction in L1 and L3 cache miss rate during querying, and up to 83.0\% and 82.6\% reduction in L1 and L3 cache miss rate during indexing, respectively.}
  % \caption{[Left: \name-BF vs BF] The impact of different sizes of \name-BF vs BF on the evaluation metrics. Query Time and indexing time of \name-BF grow significantly slower when compared with BF, while achieving similar FPR. For the same BF size, \name-BF achieves up to 46.7\% and 51.9\% reduction in query time and indexing time, respectively. [Right: \name-ABF vs ABF] The impact of different sizes of \name-ABF vs ABF on the evaluation metrics. For the same size, \name-ABF achieves up to 33.1\% and 28.6\% reduction in query time and indexing time, respectively.}
    \vspace{-2.5mm}
  \label{fig:single_bf_and_abf_size}
\end{figure*}
Consider a standard partitioned BF with total range $m$, $\eta$ independent hash repetitions in separate range (i.e., the range of each hash is $m' = m /\eta$). Consider an \name function to be used with BF constructed using Theorem~\ref{thm:genconstruction}. For a single repetition, let the LSH function be $\phi$, and two RH functions be $\rho_1$ with range $[m']$ and $\rho_2$ with range $[L]$. Let the probability of collision between two tokens under $\phi$ be denoted by $\zeta(x,y)$. Consider a sequence of tokens $\{x_i\}_{i=1}^n$ be inserted into the \namens-BF and let the query token be $q$.
We make the following assumptions,
\begin{itemize}[nosep, leftmargin=*]
    \item \textbf{Assumption 1:} $\zeta(x_i, x_j) = 0$ if $|i - j |\geq w_1$ for some $w_1 \in \mathbf{N}$
    
    \item \textbf{Assumption 2:} $|\{x_i | \zeta(q, x_i) > 0 \}| \leq w_2$
\end{itemize}
The first assumption states that the design of the LSH function on the sequence is such that only nearby tokens are similar and faraway tokens are completely dissimilar. The second assumption states that, at a time, the query token can only have similarities with a few tokens of the data. Under these conditions, the false positive rate of the \namens-BF can be bounded, and the result is presented in the theorem below,
\begin{theorem} \label{thm:main}
Consider a standard partition BF with total range $m$ is constructed on sequence of tokens $\{x_i\}_{i=1}^n$ using $\eta$ independent \name functions constructed as per Theorem~\ref{thm:genconstruction} with LSH function $\phi$ and two RH functions $\rho_1$ with range $[\frac{m}{\eta}]$ and $\rho_2$ with range $[L]$. The false positive rate $\epsilon$ for a query token $q$ under assumptions 1 and 2 is bounded as follows,
\begin{align*}
\epsilon  \leq \left( w_2\left(\frac{1}{L} + \frac{\eta}{m} \right)   + 2 \left( 1 -  \left( 1 - \left( \frac{w_1\eta}{m} \right) \right)^{\frac{n}{2w_1}} \right) \right)^\eta\\
\approx \left( w_2\left(\frac{1}{L} + \frac{\eta}{m} \right)   + 2 \left( 1 -  e^{-\frac{\eta n}{2m}}  \right) \right)^\eta
\end{align*}
\end{theorem}
We make the following observations from the theorem.
For a given value of $\eta$, even if $m \rightarrow \infty$, the RHS is bounded by $(w_2/L)^\eta$. This is expected since $L$ controls the probability of collision among highly similar tokens. Also, for a given range $m$, keeping all other parameters constant, one can optimize $\eta$ for the lowest bound. While closed form solution to the minimization problem is difficult, a simple grid search would give us the best $\eta$. As this is the upper bound, this does not directly give us a way to choose optimal $\eta$ for our purpose. The bound is loose. In practice, we obtain much lower false positive rates, as we will see in Section~\ref{sec:exp}.

\textbf{FPR for gene-search:}
Let us check if assumptions 1 and 2 made in the previous subsection hold for genomic data. 
Table~\ref{tab:assumption1} verifies the first assumption for gene-search data. With high probability, the \textit{Jaccard similarity} on sub-kmer sets of faraway k-mers is 0. Thus, assumption 1 is valid for gene search. Under the validity of assumption 1 and the specific usage of \textit{MinHash} on sub-kmers for the LSH function involved in \namens, we can prove assumption 2 
\begin{lemma}
In gene search, while using BF with k-mer size $k$ and sub-kmer size $t$, given assumption 1 is valid with $w_1=k$, the assumption 2 is true with $w_2 = (k-t+1)^2$
\end{lemma}
The proof is deferred to the appendix. Thus both the assumptions hold for gene sequence search, and the Theorem~\ref{thm:main} is applicable to gene search with $w_1 = k$ and $w_2 = (k - t + 1)^2$
\section{Experiments} \label{sec:exp}
 In this section, we present extensive experimental results to validate the effectiveness of \name for the application of gene sequence search. We integrate IDL into BF and existing gene sequence indices- COBS and RAMBO. Subsequently, we compare their efficiency and search quality with the baselines which uses RH functions. For COBS and RAMBO, we consider both scenarios of data residing on RAM and on disk, to reflect the real-world use cases.

% In this section, we demonstrate the effectiveness of \namens-BF and \name-ABF over the baselines BF and ABF respectively for gene sequence search. 
% We perform the experiments on RAM measuring index, query times, and false positive rates (FPR). Additionally, we also perform experiments for inset and query on disk. 

\begin{figure*}
  \centering
  \includegraphics[width=\textwidth]{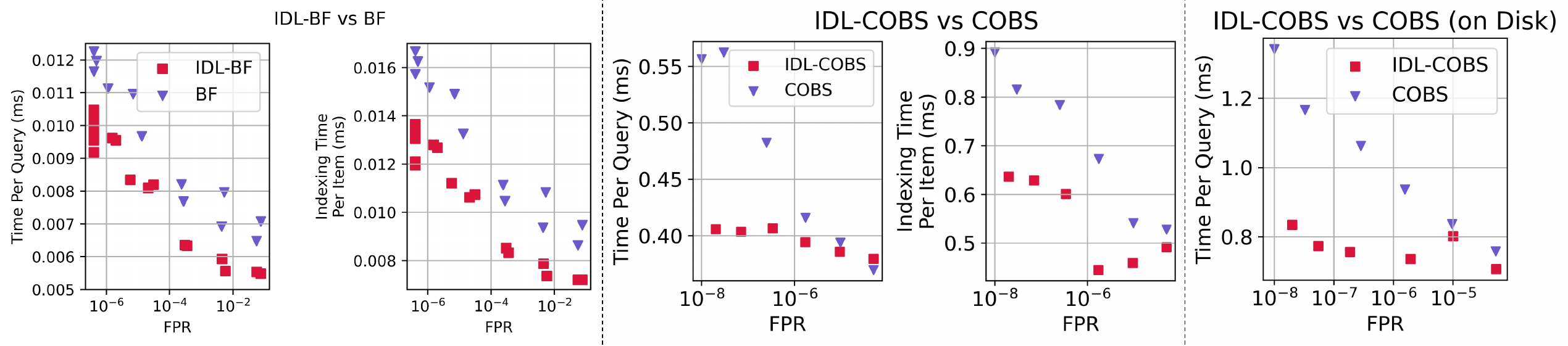}
  \caption{ Scatter plots for various configurations of \namens-BF/\namens-COBS and BF/COBS to compare index time and query time against false positive rate achieved. The best points ( lowest for given FPR) should be considered for comparison. [Left 2: \namens-BF vs BF] \namens-BF reduces query time by up to $18.1\%$ and indexing time by up to $23.0\%$ while achieving the same or better FPR as the baseline. [Middle 2: \namens-COBS vs COBS] \namens-COBS reduces query time by up to $27.7\%$ and indexing time by up to $33.8\%$ while achieving the same or better FPR the baseline. [Right 1: \namens-COBS vs COBS (on disk)] \namens-COBS on disk reduces query time by up to $28.4\%$ while achieving the same or better FPR as the baseline.}
    \vspace{-2.5mm}
  \label{fig:all_vs_fpr}
\end{figure*}

\textbf{Dataset:} We use randomly selected FASTQ\cite{FASTQ} sequence files microbial genome database
\cite{bigSI}. 
Data is downloaded using the SRAToolkit\cite{SRAToolkit} provided by NCBI. These files are all included in the European Nucleotide Archive (ENA) \cite{ENA2022}. For single-set MT, we use a randomly sampled file from the database (1.1GB, 420M kmers), while for MSMT, we use 10 randomly sampled files (27.9GB, 10.3B total kmers) for COBS and 100 files (105GB, 42B total kmers) for RAMBO. The database of gene sequences does not provide a query set of sequences. Therefore, we generate the queries with a 1-poisoning attack: for each sequence in a file, we randomly sample a subsequence of length greater than $31$ and poison the subsequence by changing one character at a random location. This process generates difficult queries for set membership testing since each query will maximally resemble inserted sequences.

\textbf{Hardware and software implementation:}
All experiments are performed on a Ubuntu 20.04 machine equipped with 2 AMD EPYC 7742 64-core processors (the L1, L2, L3 cache sizes are 2MB, 32MB, and 256MB, respectively) and 1.48TB of RAM. The code is written in C++ and compiled using the GNU Compiler with \texttt{-O3} optimization, and runs in a single thread.

\textbf{Baseline methods:} For \namens-based data structures for gene sequence search, namely \namens-BF, \namens-COBS, and \namens-RAMBO, we use the RH function based counterparts as the baseline methods. We use the general purpose \textit{MurmurHash}\cite{murmur} as the choice of random hash function. The efficiency and hashing quality of \textit{MurmurHash} make it a strong baseline.

\textbf{Evaluation metrics:} We use three evaluation metrics: \textbf{Index Time}: is the insertion time per sequence read ($~200$ kmers) averaged over all sequence reads in the database. \textbf{Query Time:} is the average query time . All query times are reported on single thread. \textbf{False positive rates (FPR):} We use FPR for MT and MSMT as defined in Section \ref{subsec:genesearch}. \textbf{Cache Miss Rates:} is the proportion of memory accesses that result in a cache miss. We report the number of cache misses simulated via Valgrind tool \cite{Valgrind} using two level model with level 1 and level 2 capacities set to 2MB (L1) and 256 MB (L3) respectively.

\subsection{Results for \namens-BF}
Figure \ref{fig:single_bf_and_abf_size} and \ref{fig:all_vs_fpr} [Left] show the experimental results of \namens-BF vs. BF. In Figure \ref{fig:single_bf_and_abf_size}, we increase the size of BF to investigate its effects on the evaluation metrics. We make the following observations.
\begin{itemize}[leftmargin=*]
    \item \textbf{Low FPR:} As the size of BF grows, the FPR is expected to decrease because distinct items are less likely to be hashed to the same location, leading to fewer hash collisions. We observe that the decrease in FPR is similar between \namens-BF and BF, which implies \name has the desirable property of having similarly low collision rates like \textit{MurmurHash}. \namens-BF has a slightly higher FPR than vanilla BF, which is expected because \name preserves locality. This slightly worse FPR is compensated by the fact that \namens-based BF achieves significant reductions in both query time and indexing time and hence can be run at higher values of BF sizes to achieve the same FPR at latency advantage (see Figure \ref{fig:all_vs_fpr}). 
    \item \textbf{Reductions in query and indexing time:} When compared to BF, \namens-BF achieves up to 41.9\% and 44.3\% reduction in query time and indexing time, respectively. Moreover, we note that the growth rates of query time and indexing time of \namens-BF are significantly slower than BF. This is accounted for by the much lower L1 cache miss rates of \namens-BF and its much lower growth rate of L3 cache miss rates.
    \item \textbf{Significant reductions in cache miss rates:} L1 cache miss rates of \namens-BF are consistently much lower than those of BF, while L3 cache miss rates of \namens-BF grow significantly slower than those of BF as the size of BF increases. For the same BF size, \namens-BF achieves up to 76.2\% and 77.0\% reduction in L1 and L3 cache miss rate during querying, and up to 83.0\% and 82.6\% reduction in L1 and L3 cache miss rate during indexing, respectively.
\end{itemize}

The gains of \namens-BF in the evaluation metrics are clear with varying BF sizes. But what are the absolute gains in query and indexing efficiency of \namens-BF while achieving the same or better search quality as baseline? The left two plots in Figure \ref{fig:all_vs_fpr} show the querying and indexing time reductions of \namens-BF over BF while achieving the same or better FPR. \namens-BF consistently achieves lower querying and indexing time over BF for the same FPR. Given the same FPR, \namens-BF reduces query time by up to 18.1\% and indexing time by up to 23.0\%, respectively.
% For almost all FPR, \namens-BF and \name-ABF are able to achieve lower querying and indexing time over baselines. We highlight that \name-ABF achieves up to $27.7\%$ reduction in query time and up to $33.8\%$ reduction in indexing time over baselines while achieving better FPR. With the low FPR and significant reductions in query and indexing time combined, our proposed algorithms translate to significant efficiency gains in real retrieval applications while maintaining search quality. In applications such as gene search, terabytes of data need to be indexed and queried efficiently, \name-based BF will have a tremendous impact on the efficiency of gene search.

\subsection{Results for \namens-COBS}
\begin{figure*}[t]
  % \centering
  % \begin{subfigure}[]{0.45\textwidth}
  % \includegraphics[scale=0.45]{imgs/abf_disk_vs_size_new.png}
  %  \end{subfigure}
  % \begin{subfigure}[]{0.45\textwidth}
  % \centering
  % \includegraphics[scale=0.45]{imgs/ablation.png}
  % \label{fig:ablation}
  %  \end{subfigure}
  %  \vspace{-3mm}
  %  \caption{[Left: \name-ABF vs ABF on Disk] The impact of different total sizes of \name-ABF vs ABF on query time and indexing time, when using disk. The growth of query time for \name-ABF is almost flat thanks to its cache efficient design, while ABF grows much faster. \name-ABF on disk achieves up to $44.32\%$ reduction in query time than baseline for the same total ABF size. [Right: Ablation Study] The effects of different values of BF size $m$, number of hash functions $\eta$, length of sub-kmer size $t$, and the random hash range $L$ of \name on FPR and query time. }
  %    \vspace{-0.2cm}

  \includegraphics[width=\textwidth]{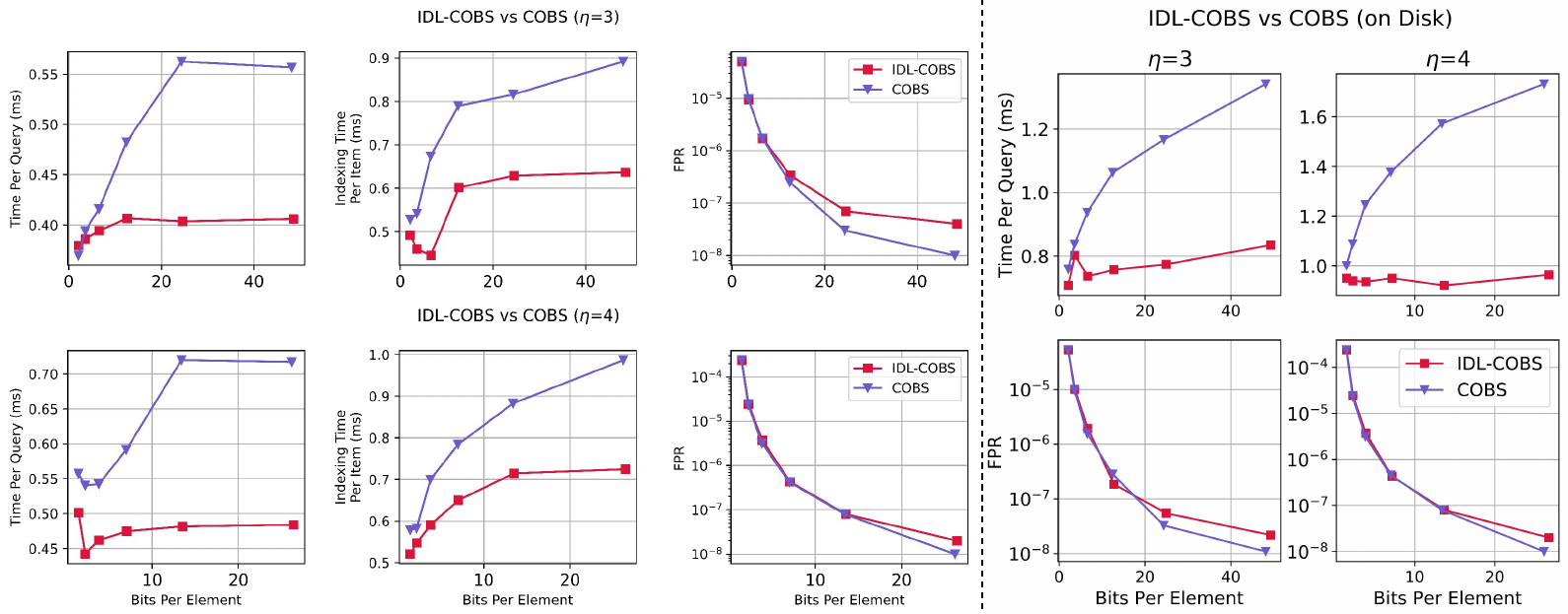}
  \caption{[Left: \namens-COBS vs COBS] The impact of different sizes of \namens-COBS vs COBS on the evaluation metrics. For the same size, \namens-COBS achieves up to 33.1\% and 28.6\% reduction in query time and indexing time, respectively. [Right: \namens-COBS vs COBS on Disk] The impact of different total sizes of \namens-COBS vs COBS on query time and indexing time, when using disk. The growth of query time for \namens-COBS is almost flat thanks to its cache efficient design, while COBS grows much faster. \namens-COBS on disk achieves up to $44.32\%$ reduction in query time than baseline for the same total COBS size.}
\label{fig:cobs}
\vspace{-2.5mm}
\end{figure*}

We perform MSMT for gene search by indexing 10 FASTQ files of size 27.9GB and 10.3 billion kmers, with the COBS index. 
% Figures \ref{fig:all_vs_fpr} and \ref{fig:cobs} shows the experimental results of \namens-COBS vs. COBS.
In Figure \ref{fig:cobs}, we increase the total size of all BFs in COBS to investigate its effects on the evaluation metrics. The results for data residing on RAM are shown on the left of Figure \ref{fig:cobs}, and the results for data residing on disk are shown on the right. Indexing time is omitted from disk usage since indexing is only done on RAM in real use cases. For the small BF ranges, \namens-COBS and COBS have similar query time and indexing time. However, as the total size grows, the query time of COBS quickly outgrows \namens-COBS, and the query time of \namens-COBS remains relatively flat. At the largest size considered, \namens-COBS on RAM achieves up to 33.1\% and 28.6\% reduction in query time and indexing time, and \namens-COBS on disk achieves up to 44.32\% reduction in query time, while maintaining similar FPR. The right two plots in Figure \ref{fig:all_vs_fpr} shows the direct comparison of query time and indexing time between COBS and \namens-COBS for the same FPR. \namens-COBS is able to outperform COBS in terms of query and indexing efficiency in almost all regimes of FPR. \namens-COBS on RAM reduces query time and indexing time by up to 27.7\% and 33.8\%, respectively, and \namens-COBS on disk reduces query time by up to 28.4\% for the same FPR.

\begin{table}[]
\centering
\caption{Query and index time improvements on the SOTA gene sequence search algorithm RAMBO \cite{gupta2021fast}. 
% We Index 100 gene sequence files, with fixed RAMBO parameters (B=20 and R=2) and varying BF sizes from 500 Million to 2000 Million.
RAMBO with IDL hash is called IDL-RAMBO. We experiment with IDL random hash range (L) of 2k (2048) and 4k (4096).  }
\label{tab:ramboNumbers}
\fontsize{8}{10}\selectfont
\setlength{\tabcolsep}{3.5pt}
\begin{tabular}{|c|c|c|c|c|c|}
\hline
 & { \textbf{BF size ($m$)}}  & { \textbf{500M}} & { \textbf{1000M}} & { \textbf{1500M}} & { \textbf{2000M}} \\
 \hline
& \textbf{(Index size)} & { 2.4 GB}   & { 4.7 GB}    & { 7 GB}    & { 9.4 GB}   \\ 
 \hline

\textbf{FPR}  & { {RAMBO}} & { 4.9E-2}& { 2.6E-3}  & { 6.3E-4}& { 2.6E-4}\\ 
  & { {IDL-RAMBO (2k)}}    & { 4.5E-2}& { 2.2E-3} & { 6.0E-4}& { 3.5E-4}\\
   & { {IDL-RAMBO (4k)}}    & { 4.6E-2}& { 2.36E-3} & { 6.0E-4}& { 3.0E-4}\\
\hline
 \textbf{Query}  & { {RAMBO}}  & { 0.138}    & { 0.11}  & { 0.117} & { 0.15}   \\
  \textbf{time (ms)}  & { {IDL-RAMBO (2k)}}     & { 0.10}    & { 0.08}  & { 0.08} & { 0.083}  \\
  \textbf{(Disk)} & { {IDL-RAMBO (4k)}}    & { 0.098}    & { 0.06}  & { 0.068}  & { 0.081}  \\
\hline
\textbf{Query} & { {RAMBO}} & { 0.101}& { 0.07} & { 0.073}   & { 0.092}  \\
 \textbf{time (ms)} & { {IDL-RAMBO (2k)}}    & { 0.07}& { 0.041} & { 0.038}  & { 0.039}  \\
 \textbf{(RAM)} & { {IDL-RAMBO (4k)}}    & { 0.078}& { 0.047}  & { 0.041}   & { 0.042}   \\
\hline
 \textbf{Index}   & { {RAMBO}}   & { 6.31}   & { 8.18}    & { 9.11}    & { 7.93}    \\
\textbf{time}  & { {IDL-RAMBO (2k)}}      & { 5.44}   & { 6.5}     & { 6.79}    & { 6.68}    \\
\textbf{(min)} & { {IDL-RAMBO (4k)}}     & { 4.61}   & { 5.12}    & { 5.12}    & { 5.645}   \\ 
\hline

\end{tabular}
\vspace{-2.5mm}
\end{table}

With the low FPR and significant reductions in query and indexing time combined, our proposed algorithms translate to significant efficiency gains in real retrieval applications while maintaining search quality. In applications such as gene search, terabytes of data need to be indexed and queried efficiently, \name hash functions will have a tremendous impact on the efficiency of gene search.

\subsection{\namens-RAMBO for large-scale gene search}
 The RAMBO \cite{gupta2021fast} index is designed for large scale data. Hence, we take a step further to index 100 files on RAMBO index. For this experiment, we use number of repetitions R=2 and number of BF in each repetitions B to be 20.  With this, RAMBO has a total of 40 BFs for 100 files. The total size of the data to index is 105GB, and the total number of basepairs are approximately 49 Billion (42 Billion kmers). Table \ref{tab:ramboNumbers}, shows the index size, FPR, query time with index on disk and RAM, and total index time with varying BF size on RAMBO and the proposed \namens-RAMBO index. We vary the BF sizes from 500 Million to 2000 Million bits and the number of hash $\eta$ is 4. The query is performed over single thread and the indexing is done with OpenMP \cite{dagum1998openmp} shared memory parallelism on 64 threads. The \name hash plugin over RAMBO provides the similar FPRs with up to \textbf{$2.2 \times$} query speedup and \textbf{$1.7\times$} index time speedup. 

\subsection{Ablation study}
\begin{figure}[t]
  \includegraphics[width=\linewidth]{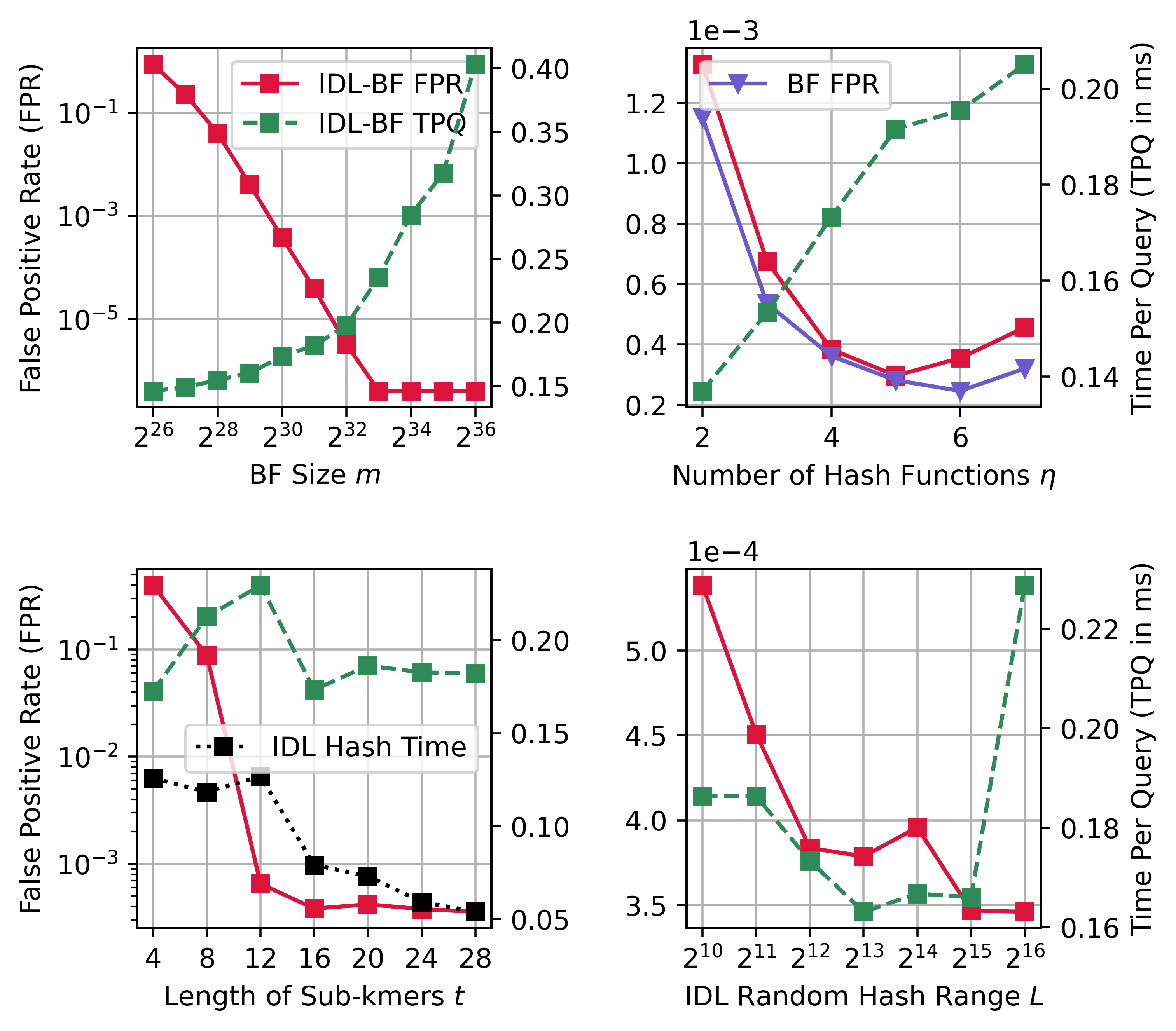}
   % \caption{[Left: \name-ABF vs ABF on Disk] The impact of different total sizes of \name-ABF vs ABF on query time and indexing time, when using disk. The growth of query time for \name-ABF is almost flat thanks to its cache efficient design, while ABF grows much faster. \name-ABF on disk achieves up to $44.32\%$ reduction in query time than baseline for the same total ABF size. [Right: Ablation Study] The effects of different values of BF size $m$, number of hash functions $\eta$, length of sub-kmer size $t$, and the random hash range $L$ of \name on FPR and query time. }
   \caption{The effects of different values of BF size $m$, number of hash functions $\eta$, length of sub-kmer size $t$, and the random hash range $L$ of \name on FPR and query time.}
\label{fig:ablation}
\vspace{-2mm}
\end{figure}

In this ablation study, we perform experiments to study the effects of all the free parameters of \name on FPR and query time of \namens-BF. We list the tunable parameters of \namens-BF as follows:
\begin{itemize}[nosep, leftmargin=*]
    \item the number of bits $m$ in the array underlying the BF,
    \item the number of \name hash functions $\eta$ used for the BF,
    \item the length of sub-kmers $t$,
    \item and the range $L$ of RH function in \namens.
\end{itemize}
We choose a set of parameters for \namens-BF that achieves less than $10^{-3}$ FPR while achieving good efficiency as the base parameters and vary the free parameters to study the effects they have on FPR and query time. The results are shown in figure \ref{fig:ablation}. We make the following observations.

\begin{itemize}[leftmargin=*]
    \item \textbf{The trend of FPR:} The FPR is expected to decrease with increasing $m, t$ or $L$, while FPR initially decreases then increases as $\eta$ increases. With increasing BF size $m$ or \name random hash range $L$, the hash collision rate decreases hence reducing the FPR. With increasing sub-kmer length $t$, two random kmers are less likely to share common sub-kmers, effectively decreasing the locality property of \name. On the other hand, increasing $\eta$ decreases the probability of two items sharing the same hash codes but increases the number of bits set to $1$ in the BF.
    \item \textbf{The trend of query time with $m$ and $\eta$:} The query time is expected to increase monotonically with $m$ and $\eta$. Increasing the BF size $m$ increases the number of cache misses and page faults, leading to higher query time. Increasing the number of hash functions $\eta$ increases the amount of hash calls and check operations, hence increasing the query time.
    \item \textbf{High value of $L$ leads to page faults:} By using values of $L$ less than $2^{16}$, \namens-BF achieves similar query time, but there is a significant increase in query time when going from $L=2^{15}$ to $L=2^{16}$. This can be interpreted as more page faults during querying. Since the page size in our system is $2^{15}$ bits, increasing $L$ beyond that will result in significantly more page faults, leading to much higher query time. Therefore, we recommend using values of $L$ at or slightly below the page size to achieve the best search quality and efficiency.
    \item \textbf{The trend of hash time with $t$:} The length of sub-kmers $t$ heavily affects the hash time of \namens. Going from $t=12$ to $t=16$ leads to a significant drop in hash time, which in turn decreases query time. This is due to the use of a segment tree in our implementation for efficient computation of the minimum operation (see Algorithm~\ref{alg:rollingmh}). Increasing $t$ from $12$ to $16$ reduces the size of the segment tree by more than half, thus reducing the amount of computation during hash operations. Therefore, we recommend using $t=16$ for a kmer size of $31$ to strike a balance between cache efficiency and low computational overhead.
    \item \textbf{Choice of $\eta$:} We compare the FPR of \namens-BF with BF with increasing number of hash functions $\eta$ increases in the top right plot of Figure \ref{fig:ablation}. The two curves are closely aligned, implying \namens-BF and BF behave similarly with the same choice of $\eta$. Therefore, we recommend setting $\eta$ of \namens-BF to the optimal value of $\eta$ in the BF using a RH function and trying adjacent values of $\eta$ for parameter tuning.
\end{itemize}

\subsubsection{Replacing RH with LSH for gene search}

This section explores the impact of using LSH functions as a substitute for RH functions in gene search data structures. Specifically, we evaluate \textit{MinHash}, a type of LSH function, and compare it against RH and \name with regards to query time, cache misses, and FPR when utilized in a BF. Table \ref{tab:minhash} presents the results for three different sizes of BF for the same gene search task.

While \textit{MinHash} demonstrates superior cache efficiency in terms of lower query time and cache miss rates compared to RH and \namens, its high FPR makes it unsuitable for gene search. The high FPR associated with LSH is attributed to the loss of input identity due to more collisions. In contrast, our proposed approach \namens, which combines both LSH and RH, provides the best balance between cache efficiency and search quality.
\begin{table}[]
\caption{A comparison of \textit{MinHash} with RH and IDL as the hash function of BF for gene search. \textit{MinHash} exhibits the best cache efficiency out of the three hash functions, shown by the lowest query time and cache miss rates, but at the cost of much worse FPR, due to a high number of hash collisions.}
\label{tab:minhash}
\centering
\small
\setlength{\tabcolsep}{4pt}
\begin{tabular}{|r|l|r|r|r|r|}
\hline
\multicolumn{1}{|l|}{\multirow{2}{*}{\textbf{BF Size}}} & \multirow{2}{*}{\textbf{Hash}} & \multicolumn{1}{l|}{\textbf{Time Per}}   & \multicolumn{1}{l|}{\textbf{L1 Cache}}  & \multicolumn{1}{l|}{\textbf{L3 Cache}}  & \multicolumn{1}{l|}{\multirow{2}{*}{\textbf{FPR}}} \\
\multicolumn{1}{|l|}{}                         &                       & \multicolumn{1}{l|}{\textbf{Query (ms)}} & \multicolumn{1}{l|}{\textbf{Miss Rate}} & \multicolumn{1}{l|}{\textbf{Miss Rate}} & \multicolumn{1}{l|}{}                     \\ \hline
\multirow{3}{*}{$2^{30}$}                    & MinHash               & 0.389                           & 2.40\%                         & 0.50\%                         & 7.704E-2                                  \\ \cline{2-6} 
                                               & RH                & 0.576                           & 25.19\%                        & 0.71\%                         & 3.613E-4                                  \\ \cline{2-6} 
                                               & IDL                   & 0.479                           & 13.30\%                        & 0.70\%                         & 3.836E-4                                  \\ \hline
\multirow{3}{*}{$2^{32}$}                    & MinHash               & 0.492                           & 2.39\%                         & 1.26\%                         & 7.509E-2                                  \\ \cline{2-6} 
                                               & RH                & 0.682                           & 25.14\%                        & 9.12\%                         & 1.619E-6                                  \\ \cline{2-6} 
                                               & IDL                   & 0.533                           & 13.26\%                        & 4.78\%                         & 3.237E-6                                  \\ \hline
\multirow{3}{*}{$2^{34}$}                   & MinHash               & 0.429                           & 2.39\%                         & 1.62\%                         & 7.460E-2                                  \\ \cline{2-6} 
                                               & RH                & 0.871                           & 25.13\%                        & 14.80\%                        & 0                                  \\ \cline{2-6} 
                                               & IDL                   & 0.569                           & 13.25\%                        & 7.29\%                         & 0                                  \\ \hline
\end{tabular}
   % \vspace{-4mm}
\end{table}

% \gaurav{\subsection{Comparisons with other LSH functions}}

% \vspace{-3mm}

% \input{sections/experiments.tex}
% \vspace{-1mm}
\section{Future Work}
The IDL hash is a drop in replacement of random hash and thus is easy to implement in various hashing-based systems; it can improve the retrieval efficiency of systems that exhibit similarities among temporally close queries by using LSH functions that capture the similarity. Although this paper focuses on the specific application of genome search using Bloom filters, the general recipe of IDL can be potentially extended to other information retrieval applications. For instance, IDL can be used in various large-scale industrial applications involving image/text-tokenized data where search is based on exact patch/token matching and queries show temporal correlations. Another potential application of IDL is searching through extensive system-error/network logs and music search systems. We leave the exploration of further applications of IDL as future work.
% The IDL hash function minimizes cache misses for datasets of queries with temporal correlations. Specifically, IDL BF-based retrieval models are designed for search applications that require an exact match. While time-series datasets, excluding DNA, may not align perfectly with these models since most related applications focus on similarity matches [1], we believe that IDL Bloom Filters can offer advantages in various large-scale industrial applications involving image/text-tokenized data. For instance, within a retail infrastructure, a significant portion of product searches relies on keyword matching. Furthermore, these queries often exhibit correlations within a session. IDL-BF based indices provide identity with spatial locality on the correlated products. Another potential application is searching through extensive system-error/network logs, where many keywords are temporarily co-located in both queries and the data. Although there is currently no comprehensive public dataset available for these purposes, we are actively in discussions with prominent retail and cloud infrastructure companies to acquire curated datasets.  Please note that we have not yet conducted experiments on these industry-specific use cases, henceforth we don’t claim it as a contribution in our paper.
\section{Conclusion}

Hash functions are critical for efficient data mining systems. However, currently used classes of hash functions such as random hash and \textit{locality sensitive hash} fall short of correct utilization of system architecture. This paper highlights the issue of cache/page inefficiency in gene sequence search and proposes a new class of hash functions \textit{\longname}~(\namens) that is able to retain good properties of existing hash functions while improving the system performance. Specifically, with \name we improve gene sequence search indexing and query times by a huge margin.
% \vspace{-1mm}
\section{Proofs}
% We list details on proofs of the theorems in this section.
\subsection{Proof of Theorem 1}
\textbf{Case 1:} If $d_U(x,y) \leq r_1$, then,
\begin{align*}
    & \mathbf{Pr}\Big(\mathbf{1}(\psi(x) \neq \psi(y)) \wedge d_V(\psi(x), \psi(y)) < L \Big) = \mathbf{Pr}\Big(\phi(x) = \phi(y)\Big)\frac{L-1}{L} \\
    & \qquad + \mathbf{Pr}\Big(\phi(x) \neq \phi(y) \wedge ((\mathbf{1}(\psi(x) \neq \psi(y)) \wedge d_V(\psi(x), \psi(y)) < L ))\Big)\\
    & \ \ \ \ \ \ \ \ \geq \mathbf{Pr}\Big(\phi(x) = \phi(y)\Big)\frac{L-1}{L}  = p_1 \frac{L-1}{L} 
\end{align*}
\textbf{Case 2:} If $d_U(x,y) > r_2$, then,
\begin{align*}
    & \mathbf{Pr}\Big( d_V(\psi(x), \psi(y)) < L \Big) = \mathbf{Pr}\Big(\phi(x) = \phi(y)\Big) \\
    & \; + \mathbf{Pr}\Big(\phi(x) \neq \phi(y) \wedge (|\rho_1(x) - \rho_1(y)| < L) \wedge ( (d_V(\psi(x), \psi(y)) < L)) \Big)\\
    & \leq p_2 + \mathbf{Pr}(\phi(x) \neq \phi(y) \wedge (\mathbf{Pr}(|\rho_1(x) - \rho_1(y)| < L)\\
    & \leq p_2 + (\mathbf{Pr}(|\rho_1(x) - \rho_1(y)| < L) = p_2 + \frac{L}{m}\\
\end{align*}

\subsection{Proof of lemma 1}
Any query kmer can have at most $(k-t+1)$ sub-kmers and each sub-kmer can be part of at most $(k-t+1)$ kmers in the data. Hence, there can be atmost $(k-t+1)^2$ kmers that have non-zero \textit{Jaccard similarity} and hence the probability of collision with query kmer.

\subsection{Proof of Theorem 2}
We make the following assumptions about the setup,
\begin{itemize}
    \item \textbf{A1:} $\zeta(x_i, x_j) = 0 \textrm{ if } |i - j| >= w_1$
    \item \textbf{A2:}$| \{x_i \; | \;  \zeta(q, x_i) > 0\}| \leq w_2$
\end{itemize}
Lemma 1 is true under these assumptions. Let us now compute false positive rate for one independent repetition.
\begin{align*}
 p_q  = \mathbf{Pr}_{h \leftarrow \mathcal{H}}\left(\bigvee_{i=1}^n \mathbf{1}(h(q) = h(x_i) \right) 
\end{align*}
Let $C = \{x_{i_j}\}_{j=1}^{w_2}$ be the tokens in the sequence which has non-zero collision probability with $q$
Also let us divide the tokens into sets of size $w_1$ ,
\begin{equation}
    X_i = \{x_{iw_i}, ... x_{(i+1)w_i}\}
\end{equation}
Let $\bar{X}_i$ be defined as
    $\bar{X}_i = \{x | x \in X_i \wedge x \notin C \}$. 
Then we can write,
\begin{align*}
 p_q & = \mathbf{Pr}_{h \leftarrow \mathcal{H}} \Bigg( \left(\bigvee_{x \in C} \mathbf{1}(h(q) = h(x) \right) 
 \vee \left(\bigvee_{x \in \cup \bar{X}_{2i}} \mathbf{1}(h(q) = h(x) \right) \\
 & \vee \left(\bigvee_{x \in \cup \Bar{X}_{2i+1}} \mathbf{1}(h(q) = h(x) \right) \Bigg)
\end{align*}
Using union bound, and using notation $\{ h(X) = h(x) | x \in X \}$ We can write,
\begin{align*}
 p_q  \leq  & \Bigg( \left( \sum_{x \in C} \left(\frac{\zeta(x, q)}{L} + \frac{\eta}{m}\right) \right)  + \mathbf{Pr}_{h \leftarrow \mathcal{H}} \left(\bigvee_{X=\bar{X}_{2i}} \mathbf{1}(h(q) \in h(X) \right) \\
 & + \mathbf{Pr}_{h \leftarrow \mathcal{H}}  \left(\bigvee_{X=\bar{X}_{2i+1}} \mathbf{1}(h(q) \in h(X) \right)\Bigg)
\end{align*}
\begin{lemma}
For distinct $i$ and $j$, the events of the type $\mathbf{1}(h(q) \in X_{2i})$ and  $\mathbf{1}(h(q) \in X_{2j})$ are independent. The events of the type $\mathbf{1}(h(q) \in X_{2i+1})$ and  $\mathbf{1}(h(q) \in X_{2j+1})$ are independent. 
\end{lemma}
Using independence and rewriting, 
\begin{align*}
 p_q  \leq  & \sum_{x \in C} \left(\frac{\zeta(x, q)}{L} + \frac{\eta}{m} \right)   + \left( 1 -  \prod_{X=\bar{X}_{2i}} \left( 1 - \mathbf{Pr}_{h \leftarrow \mathcal{H}} \left( \mathbf{1}(h(q) \in h(X) \right) \right) \right)  \\
 & + \left( 1 -  \prod_{X=\bar{X}_{2i+1}} \left( 1 - \mathbf{Pr}_{h \leftarrow \mathcal{H}} \left( \mathbf{1}(h(q) \in h(X) \right) \right) \right)
\end{align*}
Applying union bound within sets and by replacing $\bar{X}_i$ with $X$, we only increase the RHS. Hence,
\begin{align*}
 p_q  \leq & \sum_{x \in C} \left(\frac{\zeta(x, q)}{L} + \frac{\eta}{m} \right)   + \left( 1 -  \prod_{X=X_{2i}} \left( 1 - \sum_{x \in X} \mathbf{Pr}_{h \leftarrow \mathcal{H}} \left( \mathbf{1}(h(q) = h(x) \right) \right) \right) \\
 & + \left( 1 -  \prod_{X=X_{2i+1}} \left( 1 - \sum_{x \in X} \mathbf{Pr}_{h \leftarrow \mathcal{H}} \left( \mathbf{1}(h(q) = h(x) \right) \right) \right) 
\end{align*}
\begin{align*}
 p_q  & \leq  \sum_{x \in C} \left(\frac{\zeta(x, q)}{L} + \frac{\eta}{m} \right)   + \left( 1 -  \prod_{X=X_{2i}} \left( 1 - \sum_{x \in X} \left( \frac{\eta}{m} \right) \right) \right)   \\
 & + \left( 1 -  \prod_{X=X_{2i+1}} \left( 1 - \sum_{x \in X} \left( \frac{\eta}{m} \right) \right) \right) 
\end{align*}
With equation crunching and using $\zeta(x,q) \leq 1$, we can simplify,
\begin{align*}
 p_q  \leq  w_2\left(\frac{1}{L} + \frac{\eta}{m} \right)   + 2 \left( 1 -  \left( 1 - \left( \frac{w_1\eta}{m} \right) \right)^{\frac{n}{2w_1}} \right) 
\end{align*}
Thus the false positive rates are bounded by
\begin{align*}
\epsilon  \leq & \left( w_2\left(\frac{1}{L} + \frac{\eta}{m} \right)   + 2 \left( 1 -  \left( 1 - \left( \frac{w_1\eta}{m} \right) \right)^{\frac{n}{2w_1}} \right) \right)^\eta \\ 
& \approx \left( w_2\left(\frac{1}{L} + \frac{\eta}{m} \right)   + 2 \left( 1 -  e^{-\frac{\eta n}{2m}}  \right) \right)^\eta
\end{align*}
Note for a given $\eta$, the if $m\rightarrow \infty$, the $\epsilon$ is upper bounded by $(\frac{w_2}{L})^\eta$.  For large enough $L$ and a reasonable value of $\eta$, this bound is also very small.

% Identification of funding sources and other support, and thanks to
% individuals and groups that assisted in the research and the
% preparation of the work should be included in an acknowledgment
% section, which is placed just before the reference section in your
% document.

\bibliographystyle{ACM-Reference-Format}
\bibliography{ref.bib}

\end{document}